\documentclass[letterpaper, 10 pt, conference]{ieeeconf}
\IEEEoverridecommandlockouts     
\usepackage{cite}
\usepackage{amsmath,amssymb,amsfonts}

\usepackage{amsthm}
\usepackage{comment}
\newtheorem{lem}{Lemma}
\newtheorem{thm}{Theorem}

\theoremstyle{definition}

\newtheorem{assumption}{Assumption}

\usepackage{caption}  %
\usepackage{subcaption} %

\usepackage{amsmath, amssymb, amsfonts}
\usepackage{xcolor}
\newcommand{\boxend}{\hfill \ensuremath{\Box}}

\usepackage{algorithm}
\usepackage{algpseudocode}

\newcommand{\FedLQR}{\texttt{FedLQR}~}

\usepackage{xcolor}

\usepackage{bm}             
\usepackage{dsfont}         

\usepackage{mathtools}       
\usepackage{array}

\usepackage{graphicx}
\usepackage{textcomp}
\def\BibTeX{{\rm B\kern-.05em{\sc i\kern-.025em b}\kern-.08em
    T\kern-.1667em\lower.7ex\hbox{E}\kern-.125emX}}
\begin{document}
\title{Scalar Federated Learning for Linear Quadratic Regulator}
\author{Mohammadreza Rostami, Shahriar Talebi, \emph{Member, IEEE} and Solmaz S. Kia, \emph{Senior Member, IEEE}
\thanks{This work is partially supported by the 2026 UCLA–UCI PI grant awarded by the Samueli Foundation (Corresponding author: S. Talebi).}
\thanks{M. Rostami and S. Kia are with the University of California, Irvine, CA, US, and S. Talebi is with the University of California, Los Angeles, CA, US. Emails: {\em\small mrostam2@uci.edu}, {\em\small s.talebi@ucla.edu} and {\em\small solmaz@uci.edu} }
}

\maketitle

\begin{abstract}
We propose \textsc{ScalarFedLQR}, a communication-efficient federated algorithm for model-free learning of a \emph{common policy} in linear quadratic regulator (LQR) control of cooperative agents. The method builds on a decomposed projected gradient mechanism, in which each agent communicates only a scalar projection of a local zeroth-order gradient estimate. The server aggregates these scalar messages to reconstruct a global descent direction, reducing per-agent uplink communication from $\mathcal{O}(d)$ to $\mathcal{O}(1)$, independent of the policy dimension.
Crucially, the projection-induced approximation error diminishes as the number of participating agents increases, yielding a favorable scaling law: larger fleets enable more accurate gradient recovery, admit larger stepsizes, and achieve faster linear convergence despite high dimensionality. Under standard regularity conditions for homogeneous agents, all iterates remain stabilizing and the average LQR cost decreases linearly fast. Numerical results demonstrate performance comparable to full-gradient federated LQR with substantially reduced communication, even for heterogeneous (but similar) agents.
\end{abstract}

\section{Introduction}\fontsize{10}{11.1}\rm

Model-based design has long been the foundation of feedback control, but as
system complexity grows, developing and maintaining accurate first-principles
models becomes increasingly difficult. This has motivated data-driven approaches
that learn control policies directly from operational trajectories. Among
existing approaches, policy optimization (PO) is a promising paradigm for
data-driven control~\cite{hu_toward_2023, talebi_policy_2026}, as despite
nonconvexity, policy gradient (PG) methods enjoy global convergence in
structured settings such as LQR~\cite{fazel_global_2018,bu_lqr_2019}. However,
these methods learn from data generated by the physical system itself, and that
process is costly.

The first cost is the \emph{physical price of gradient estimation}. In
zeroth-order (ZO) model-free LQR~\cite{fazel_global_2018,
malik_derivative-free_2019}, estimating the policy gradient at the current
gain requires evaluating the cost under randomly perturbed policies using
trajectory rollouts. Depending on the ZO oracle, each perturbation direction
may involve either a single perturbed-policy evaluation, as in one-point
estimators, or a pair of evaluations under oppositely perturbed policies, as
in two-point estimators. These cost evaluations are combined across multiple
random perturbation directions to construct an approximate policy gradient,
with the required sampling effort increasing as higher estimation accuracy is
demanded~\cite{dean_sample_2020}. Crucially, these oracle evaluations are not
cheap computations: each policy evaluation requires a trajectory rollout of
length $\tau$ on the physical system~\cite{dulac-arnold_challenges_2019}.
Consequently, increasing the number of perturbation directions or the number
of evaluations per direction translates directly into additional physical
interaction with the system. A drone must interrupt its mission and expend
battery, a grid controller applies perturbations that stress equipment, and a
robotic arm incurs production downtime. Sample efficiency is therefore a
deployment requirement rather than a theoretical
nicety~\cite{mohammadi_learning_2020,zhao_global_2023,
talebi_data-driven_2023,umenberger_globally_2022}.

Large-scale multi-agent systems offer a natural remedy. {When $M$ homogeneous agents pool trajectory data, the per-agent sampling
burden decreases by a factor of $M$.} Recent work exploits this cooperative structure, including
\texttt{D2SPI}~\cite{alemzadeh_data-driven_2024} and
\texttt{FedLQR}~\cite{wang_model-free_2023}. As shown by Wang
et al.~\cite{wang_model-free_2023}, the requirement improves from
$\mathcal{O}(1/\epsilon^2)$ to $\mathcal{O}(1/(M\epsilon^2))$ by exploiting
independent gradient noise across agents while optimizing average fleet
performance. The residual burden of requiring all $M$ agents to explore at every
round can be further relaxed by client subsampling, precisely because the
per-agent complexity scales as
$\mathcal{O}(1/(M\epsilon^2))$~\cite{qi_federated_2021}.

Cooperation, however, introduces a second cost that is often just as binding,
namely communication. In \texttt{FedLQR} each agent transmits a full gradient
matrix $\Delta^{(i)}\in\mathbb{R}^{n_u\times n_x}$, incurring $\mathcal{O}(d)$
uplink cost per agent and $\mathcal{O}(Md)$ total server traffic, where
$d=n_un_x$. This payload grows with both fleet size and system
dimension~\cite{mcmahan_communication-efficient_2017, konevcny_federated_2016},
and the underlying asymmetry is fundamental. The server is provisioned and
grid-powered, whereas each agent pays for its uplink in energy and airtime it
must otherwise spend on its mission.
{In federated settings, one line of work trades extra local computation for
fewer communication rounds, using multiple local updates and minibatched
rollouts~\cite{kanakeri2025harnessing}. We take a complementary and more direct
route, reducing the size of each uplink message rather than the number of rounds.}
This is particularly effective on bandwidth-limited devices, where large messages
incur upload latency that translates back into energy cost, since a platform such
as a drone continues to burn battery while idling between rollouts. Transmitting
raw gradients may further expose local dynamics through gradient inversion
attacks~\cite{zhu_deep_2019}.

We propose \textsc{ScalarFedLQR}, which compresses the uplink via \emph{projected
directional derivatives}. Rather than transmitting $\widehat{\nabla}J(K) \in
\mathbb{R}^d$, each participating agent draws a Rademacher direction $v \in
\{\pm1\}^d$ from a pseudorandom seed and sends only the scalar projection
$\langle v, \widehat{\nabla}J(K)\rangle$ together with that seed. The server
regenerates each $v$ deterministically and aggregates the scalar messages into a
global descent direction. This reduces per-agent communication from
$\mathcal{O}(d)$ to $\mathcal{O}(1)$ and total server cost from
$\mathcal{O}(Md)$ to $\mathcal{O}(M)$, independent of system dimension, while
never exposing a full local gradient. The resulting error decomposes into
projection distortion and ZO estimation noise, jointly governed by sample size,
dimension, and fleet
size~\cite{fournier2023can,rostami2024projected,nesterov2017random}.

Under standard regularity conditions on the average LQR cost, namely a
Polyak--\L{}ojasiewicz condition with constant $\mu_c$ and local Lipschitz
continuity with $L_c$ over a stabilizing sublevel set, we establish linear
convergence.

\smallskip
\noindent\textbf{Theorem. }(Linear convergence---Informal)
\emph{If the local zeroth-order relative errors are bounded by $\epsilon$ and
\[\textstyle d \, \log(2d/\delta) \, M^{-1} \,(1+\epsilon)^2 \lesssim \beta^2 \quad \text{for some } \beta \in (0,1),\]
then, with a suitable constant stepsize and probability at least $1-\delta$,
\textsc{ScalarFedLQR} achieves geometric decay of the average cost at rate
$1 - (\mu_c/L_c){(1-\beta)^2}/{(1+\beta)^2}.$\boxend}
\smallskip

\noindent This condition exposes the central scaling law of the method. Scalar
projection introduces a dimension-dependent error, but averaging across agents
attenuates it as $M$ grows. Larger fleets therefore admit smaller $\beta$, hence
larger stepsizes and faster linear convergence even in high dimensions, whereas
smaller or noisier fleets require more conservative updates. \textsc{ScalarFedLQR}
consequently delivers a compounding benefit at scale, combining constant-size messages with stability and convergence that \emph{improve} as the
fleet grows.

The rest of the paper is organized as follows.
\S {II formulates the federated model-free LQR problem and introduces the
stabilizing set.}
\S III presents the \textsc{ScalarFedLQR} algorithm, and \S IV analyzes its
stability and convergence, including high-probability bounds on the
scalar-projection error and linear convergence under a PL condition.
\S V {reports numerical experiments comparing \textsc{ScalarFedLQR} with
\texttt{FedLQR} under varying communication budgets, and \S VI concludes.}
Due to the page limit, missing proofs are provided in the extended version of
this work~\cite{MR-ST-SSK:26}.

\section{Problem Formulation and Objective}

We consider a network of $M$ agents that communicate
through a server in a federated learning paradigm. Each agent is governed
by discrete-time linear time-invariant (LTI) dynamics 
\begin{equation}\label{eq:dynamics}
x^{(n)}_{t+1} = A^{(n)} x^{(n)}_t + B^{(n)} u^{(n)}_t, ~x^{(n)}_{0}\sim \mathcal{D}, ~n \in\{ 1,\dots,M\},
\end{equation}
where $x^{(n)}_t \in \mathbb{R}^{n_x}$ and $u^{(n)}_t \in \mathbb{R}^{n_u}$ denote the state and control input of agent $n$ and $\mathcal{D}$ is the non-degenerate distribution generating the initial conditions. The system matrices $(A^{(n)}, B^{(n)})$ are unknown and may vary across agents. All agents share the same state and input dimensions but may exhibit heterogeneous dynamics.

There exists a nominal linear model $(A,B)$ such that for all agents $n$,
\begin{equation}\label{eq::system_bound}
\|A^{(n)}-A\|_2 \leq \epsilon_1, \quad\text{and}\quad \|B^{(n)}-B\|_2 \le \epsilon_2,
\end{equation}
for some heterogeneity parameters $\epsilon_1,\epsilon_2 \ge 0$.
This similarity assumption captures the setting in which agents have distinct but closely related dynamics, ensuring the existence of a meaningful common policy. {
Our theoretical analysis focuses on the homogeneous case
$\epsilon_1=\epsilon_2=0$; robustness to heterogeneous dynamics is demonstrated
empirically, while a formal heterogeneous analysis is deferred to the extended version.
}

Each agent applies a static state-feedback control law $u^{(n)}_t = -K x^{(n)}_t$, where $K \in \mathbb{R}^{n_u \times n_x}$ is a common policy gain to be learned cooperatively. Under this policy, agent $n$ incurs the infinite-horizon quadratic cost
\begin{equation*}
J^{(n)}(K) := \mathbb{E}_{x^{(n)}_0}\!\left[ \sum\nolimits_{t=0}^{\infty} x^{(n)\top}_t Q x^{(n)}_t + u^{(n)\top}_t R u^{(n)}_t \right],
\end{equation*}
where $Q \succ 0$ and $R \succ 0$ are fixed cost matrices shared by all agents. The goal of federated learning is to compute a \emph{single} policy gain $K$ that minimizes the average LQR cost
\begin{equation}
\textstyle J_{\mathrm{avg}}(K) := \frac{1}{M} \sum\nolimits_{n=1}^M J^{(n)}(K).
\end{equation}

We emphasize that this objective differs from classical distributed control, where each agent learns its own policy. Here, we instead learn a \emph{common policy} across agents, leveraging similarity in dynamics to accelerate learning through data aggregation. While such a policy is not individually optimal, it provides a robust baseline that generalizes across agents and can be locally fine-tuned if needed. This setup enables fast fleet-level learning while retaining adaptability. The central challenge, however, is stability: under heterogeneous dynamics, a policy that stabilizes one agent may destabilize another, making the design of a \emph{commonly stabilizing policy} the key difficulty in federated LQR.
We therefore define the per-agent stabilizing set
\[
  \mathcal{S}^{(n)} := \{ K : A^{(n)} - B^{(n)} K \text{ is Schur stable} \},
\]
and let $\mathcal{S} := \bigcap_{n=1}^M \mathcal{S}^{(n)}$ denote the
set of gains that stabilize all agents simultaneously. Each
$\mathcal{S}^{(n)}$ is nonempty and open whenever
$(A^{(n)}, B^{(n)})$ is stabilizable, and $\mathcal{S}$ inherits
these properties under the following simultaneous stabilization assumption.

\begin{assumption}[Initial stabilizing policy]
There exists $K_0 \in \mathcal{S}$ that stabilizes
all agents simultaneously.
\boxend \end{assumption}

This assumption is standard in policy optimization and can be
satisfied via conservative model-based design, offline analysis,
or a simple baseline controller when the dynamics are open-loop
stable. Fully online stabilization is beyond the scope of this
work. Additional topological properties of $\mathcal{S}$---particularly
under the \emph{similar dynamics} hypothesis---are of independent interest but not central to our analysis.

Given $K_0 \in \mathcal{S}$, the federated optimization problem is
\begin{equation}
  \textstyle\min_{K \in \mathcal{S}} \; J_{\mathrm{avg}}(K).
\end{equation}
We additionally impose a \emph{communication constraint}: agents
may not transmit full policy gains or high-dimensional gradient
vectors to the server and are instead restricted to a
constant-size message per round. This models realistic bandwidth
and energy limitations in large-scale multi-agent systems and is
treated as integral to the problem formulation.

Accordingly, the objective of this work is to design a federated policy optimization scheme that
(i) minimizes the average LQR cost,
(ii) maintains all iterates within the common stabilizing set $\mathcal{S}$, and
(iii) operates under a communication model in which each agent transmits only $\mathcal{O}(1)$ information per round.
In the next section, we present a federated algorithm that operates under this communication model and analyze its stability and convergence properties.

\section{ScalarFedLQR Algorithm}

We present \textsc{ScalarFedLQR} (Algorithm~\ref{alg:scalarfedlqr_inline}), a communication-efficient federated policy optimization method for LQR systems. At each round $t$, the server broadcasts the current policy gain $K_t$ to all agents. Each agent $n$ computes a local zeroth-order estimate of its policy~gradient,
\begin{align}\label{eq::g_t_n}
\tilde g_{t,n} \coloneqq \widehat{\nabla} J^{(n)}(K_t),
\end{align}
using trajectory rollouts under the current policy. Instead of transmitting the full gradient vector, each agent samples a random Rademacher direction
$v_{t,n} \in \{-1,+1\}^d$ using a locally generated random seed (independent of other randomness), normalizes it, and forms the scalar projection \footnote{More generally,
any zero-mean isotropic random direction can be used; Rademacher directions are
adopted for their favorable variance properties~\cite{rostami2025fedmpdd}.}
$
r_n^t = v_{t,n}^\top \tilde g_{t,n}.
$
The agent uploads only this scalar together with the corresponding seed.

\begin{algorithm}[t]\footnotesize
\caption{\textsc{ScalarFedLQR}: Federated LQR via Scalar Gradient Projections }
\label{alg:scalarfedlqr_inline}
\begin{algorithmic}[1]

\State \textbf{Input:} initial stabilizing gain $K_0$, learning rate $\eta$, rounds $T$

\For{each round $t = 0,1,\ldots,T-1$}
    \State Server broadcasts $K_t$ to all $n \in [M]$

    \For{each client $n \in [M]$ \textbf{in parallel}}
        \State Generate i.i.d.\ Rademacher vector $v_{t,n}\in\{-1,+1\}^d$ using seed $\xi_{t,n}$
        \State Normalization: $v_{t,n} \gets \frac{v_{t,n}}{\|v_{t,n}\|_2}$
        \State Compute local ZO gradient estimate $\tilde g_{t,n}$ at $K_t$
        \State \textbf{Encode:}\quad $r_{n}^t \gets v_{t,n}^\top \tilde g_{t,n}$ \Comment{scalar projection}
        \State Upload $r_{n}^t\in\mathbb{R}$ and seed $\xi_{t,n}$ to the server
    \EndFor

    \State $\Delta_{\mathrm{sum}} \gets \mathbf{0}_d$ \Comment{reset decoder accumulator}

    \For{each client $n \in [M]$}
        \State Server regenerates $v_{t,n}\in\{-1,+1\}^d$ using seed $\xi_{t,n}$
        \State Normalization: $v_{t,n} \gets \frac{v_{t,n}}{\|v_{t,n}\|_2}$

        \State \textbf{Decode:}\quad $\Delta_{\mathrm{sum}} \gets \Delta_{\mathrm{sum}} + r_{n}^t\, v_{t,n}$
    \EndFor

    \State \textbf{Aggregate:}\quad $\bar g_t \gets \frac{d}{M}\,\Delta_{\mathrm{sum}}$
    \State \textbf{Model update:}\quad $K_{t+1} \gets K_t - \eta\,\bar g_t$
\EndFor

\State \textbf{Output:} $K_T$

\end{algorithmic}
\end{algorithm}

On the server side, the same random directions are deterministically regenerated using the received seeds.\footnote{Alternatively, the server may generate and
broadcast the random seeds to the agents at each round.} The server then constructs an aggregated descent direction according to
\begin{align}\label{eq:g_bar}
\bar g_t
=
\frac{d}{M}
\sum\nolimits_{n=1}^M
r_n^t\, v_{t,n}
=
\frac{d}{M}
\sum\nolimits_{n=1}^M
\bigl(v_{t,n} v_{t,n}^\top\bigr)\tilde g_{t,n}.
\end{align}
The shared policy is updated via the gradient descent step
\[
K_{t+1} = K_t - \eta\, \bar g_t,
\]
where $\eta>0$ is the stepsize.

Under this protocol, each agent transmits only a single real-valued scalar and an integer-valued seed per round. As a result, the uplink communication cost per agent is $\mathcal{O}(1)$, independent of the policy dimension $d = n_u n_x$. The server-side computation scales linearly with the number of participating agents.

The following sections analyze how the approximation error introduced by scalar projection and zeroth-order estimation affects stability and convergence, and establish conditions under which the iterates produced by \textsc{ScalarFedLQR} remain stabilizing and converge to the average optimal policy.

\section{Stability analysis and convergence of \textsc{ScalarFedLQR}}
We now study the stability and convergence properties of
\textsc{ScalarFedLQR}. Our analysis focuses on the $c$-sublevel set
$\mathcal S_c:=\{K:J_{\mathrm{avg}}(K)\le c\}$, which contains $K_0$
and is contained in the common stabilizing set $\mathcal S$ given that
$Q\succ0$ and $R\succ0$. Our goal is to show that, under suitable
conditions on the stepsize and the gradient approximation error, the
server-side iterates remain stabilizing and converge linearly to a
neighborhood of the optimal average cost.

Our analysis relies on the local smoothness of $J_{\mathrm{avg}}$ and
the local PL condition of the individual LQR costs on $\mathcal S_c$.

{Using the standard regularity properties of the LQR cost established
in~\cite{bu_lqr_2019}, $J^{(n)}$ is locally smooth and satisfies the
PL condition on $\mathcal S_c$. Since the agents
are homogeneous, $J^{(n)}=J_{\mathrm{avg}}$ for all $n$. Hence, there
exist constants $L_c>0$ and $\mu_c>0$ such that, for all
$K\in\mathcal S_c$,
\begin{subequations}
\begin{align}
\|\nabla^2 J_{\mathrm{avg}}(K)\|_2
&\le L_c, \label{eq:Lip_n}\\
\frac{1}{2}\|\nabla J_{\mathrm{avg}}(K)\|_2^2
&\ge
\mu_c\bigl(J_{\mathrm{avg}}(K)-J_{\mathrm{avg}}^\star\bigr),
\label{eq:individual_pl}
\end{align}
\end{subequations}
where
$
J_{\mathrm{avg}}^\star:=\inf_{K\in\mathcal S_c}J_{\mathrm{avg}}(K).
$}

Let $\tilde g_t$ denotes the average of the local zeroth-order gradient
estimators, while $g_t$ is the exact gradient of the average cost at the current
policy $K_t$:
\begin{align}\label{eq::zo_diff}
\textstyle\tilde g_t := \frac{1}{M}\sum\nolimits_{n=1}^{M}\tilde g_{t,n},
\qquad
g_t := \nabla J_{\mathrm{avg}}(K_t).
\end{align}
The discrepancy between the server-side aggregated direction
$\bar g_t$ and the true gradient $g_t$ reflects both the scalar-projection
reconstruction error and the zeroth-order gradient estimation error. The following result gives a one-step descent guarantee for
$J_{\mathrm{avg}}$ under the scalar-projection aggregated update. In particular, it shows that descent is ensured when the total error is sufficiently small relative to $\|g_t\|_2$ and the stepsize is chosen appropriately. 
\begin{lem}[One-step descent for scalar-projection aggregated update]\label{lem:federated_descent}
Fix \(K_t\in\mathcal S_c\), let \(g_t:=\nabla J_{\mathrm{avg}}(K_t)\in\mathbb R^d\), and
consider the update \(K_{t+1}=K_t-\eta\,\bar g_t\) with \(\bar g_t\) defined in~\eqref{eq:g_bar}.
Assume \(J_{\mathrm{avg}}\) is \(L_c\)-smooth on \(\mathcal S_c\) and
\(
\|e_t^{\mathrm{tot}}\|_2\le \beta_t\|g_t\|_2
\; \text{ for some } \beta_t\in[0,1).
\) Then, for any $\eta$ satisfying
\begin{align}\label{eta_relation_max}
0<\eta< \eta_{\max} \coloneqq \frac{2(1-\beta_t)}{L_c(1+\beta_t)^2},
\end{align}
the entire line segment joining $K_t$ and $K_{t+1}$ remains in $\mathcal S_c$. Consequently, $K_{t+1}\in\mathcal S_c$, and
\begin{equation*}
J_{\mathrm{avg}}(K_{t+1})
\!\le\!
J_{\mathrm{avg}}(K_t)\!
-\!\eta
\left[
\!(1-\beta_t)-\!\frac{L_c\eta}{2}(1+\beta_t)^2
\!\right]
\!\|g_t\|_2^2 .
\end{equation*}
In particular, if \(\beta_t\le \beta<1\) uniformly over \(t\), then a sufficient uniform stepsize condition for descent is
$
0<\eta<\frac{2(1-\beta)}{L_c(1+\beta)^2}
$. Furthermore, with the choice of $\eta^\star=\eta_{\max}/2$ we obtain
\begin{equation*}
J_{\mathrm{avg}}(K_{t+1})
\le
J_{\mathrm{avg}}(K_t)
-
\frac{(1-\beta_t)\eta_{\max}}{4}
\|g_t\|_2^2 .
\end{equation*}
\end{lem}

Lemma~\ref{lem:federated_descent} shows that $J_{\mathrm{avg}}$ decreases whenever the total gradient error $e_t^{\mathrm{tot}}$ is sufficiently small relative to $\|g_t\|_2$ and the stepsize satisfies \eqref{eta_relation_max}. The parameter $\beta_t$ therefore reflects a trade-off between stability feasibility and stepsize selection: larger $\beta_t$ makes the error condition easier to satisfy but forces smaller admissible stepsizes, while smaller $\beta_t$ allows more aggressive updates but requires a more accurate aggregated gradient direction. To extend this one-step descent guarantee to global stability, it remains to control $e_t^{\mathrm{tot}}$ in the federated model-free setting. As discussed above, $e_t^{\mathrm{tot}}$ consists of two components: the zeroth-order estimation error $e_t^{\mathrm{ZO}}$, controlled by prior FedLQR analysis, and the scalar-projection reconstruction error $e_t^{\mathrm{proj}}$, which is specific to \textsc{ScalarFedLQR}.

Referring to~\eqref{eq::zo_diff}, let
$e_t^{\mathrm{ZO}}:=\tilde g_t-g_t$ and define
$\mathcal E_t^{\mathrm{ZO}}
:=\{\|e_t^{\mathrm{ZO}}\|_2\le\varepsilon\}$.
By an analogous argument to \cite[Lemma~4]{wang_model-free_2023}, conditioned on the past
iterations, this event holds with high probability under suitable
sampling parameters. We therefore introduce a bounded gradient
heterogeneity assumption across clients, and then establish high-probability
bounds for the projection error $e_t^{\mathrm{proj}}=\bar g_t-\tilde g_t$. These
results lead directly to the global stability theorem.

{Although the agents are homogeneous and hence have the same exact policy
gradient, their local ZO estimates $\tilde g_{t,n}$ need not be identical
due to independent finite-sample rollouts, random initial conditions, and
random perturbation directions. We therefore impose the following bounded
dispersion condition on the local gradient estimates.}

\begin{assumption}[Bounded gradient dispersion]
\label{ass:heterogeneity}
For each round $t$, let
$
\Delta_{t,n}:=\tilde g_{t,n}-\tilde g_t,
$ and $
\tilde g_t:=\frac{1}{M}\sum_{n=1}^M \tilde g_{t,n}.
$
Assume that there exist nonnegative quantities $\sigma_t$ and $B_t$ such that
\begin{align}
\frac{1}{M}\sum\nolimits_{n=1}^{M}\|\Delta_{t,n}\|_2^2 &\le \sigma_t^2,
\max_{1\le n\le M}\|\Delta_{t,n}\|_2 &\le B_t.
\end{align}
\end{assumption}
Such bounded gradient dispersion (gradient dissimilarity) assumptions are common in
federated optimization~\cite{rostami2023federated}. They quantify how much each
client's local gradient $\tilde g_{t,n}$ can deviate from the round average
$\tilde g_t$.

\begin{lem}[High-probability bound for the projection error $e^{\mathrm{proj}}_t$]
\label{lem:hp_eproj_general_normalized}
Fix a round $t$. For each $n=1,\dots,M$, let
$v_{t,n}\in\{\pm 1/\sqrt d\}^d$ be an i.i.d.\ normalized Rademacher vector, and let
$\tilde g_{t,n}\in\mathbb R^d$ be arbitrary.
Then, under Assumption \ref{ass:heterogeneity} and for any $\delta\in(0,1)$, with probability at least $1-\delta$
\begin{align}\label{eq:hp_eproj_general_normalized}
\|e^{\mathrm{proj}}_t\|_2
&\le
C\sqrt{\frac{(d-1)\log(2d/\delta)}{M}}
\Big(\|\tilde g_t\|_2+\sigma_t\Big)
\nonumber\\
&\quad+
C\,\frac{(d-1)\log(2d/\delta)}{M}\,(\|\tilde g_t\|_2 + B_t),
\end{align}
for an absolute constant $C>0$.
\end{lem}
We are now in a position to establish the formal stability guarantee for
\textsc{ScalarFedLQR}. Combining the one-step descent result of
Lemma~\ref{lem:federated_descent}, the bounded heterogeneity assumption (Assumption \ref{ass:heterogeneity}), the following theorem shows that
the iterates remain in the stabilizing set $\mathcal S_c$ provided that the
total gradient error is uniformly controlled relative to $\|g_t\|_2$ and the
stepsize is chosen appropriately.

\begin{thm}[Stability of \textsc{ScalarFedLQR}]
\label{thm:stability_scalar_fixed_beta}
Fix a horizon $T\ge 1$. Assume $K_0\in\mathcal S_c$. 
For each round $t$, let $\mathcal E_t^{\mathrm{ZO}}$ denote the zeroth-order
gradient accuracy event under which
$
\|e_t^{\mathrm{ZO}}\|_2\le \varepsilon.
$
Define
$
\mathcal E_T^{\mathrm{ZO}}
:=
\bigcap_{t=0}^{T-1}\mathcal E_t^{\mathrm{ZO}}.
$
Assume there exists a fixed
$\beta\in[0,1)$ such that
\begin{equation}\label{eq:fixed_beta_condition}
\varepsilon+\zeta_t \le \beta \|g_t\|_2,
\qquad \forall t=0,\dots,T-1.
\end{equation}
If, in addition, the stepsize satisfies
\(
0<\eta<\frac{2(1-\beta)}{L_c(1+\beta)^2},
\)
then, on $\mathcal E_T^{\mathrm{ZO}}$, with probability at least $1-\delta$,
\[
J_{\mathrm{avg}}(K_{t+1})\le J_{\mathrm{avg}}(K_t),
\qquad t=0,\dots,T-1.
\]
Consequently, $
K_t\in\mathcal S_c$, for $t=0,\dots,T,$ and hence all iterates up to horizon $T$ are stabilizing.\boxend
\end{thm}

\begin{proof}
Recall that
$
e_t^{\mathrm{tot}}
:=
\bar g_t-g_t
=
(\bar g_t-\tilde g_t)+(\tilde g_t-g_t)
=
e_t^{\mathrm{proj}}+e_t^{\mathrm{ZO}} .
$
{Applying Lemma~\ref{lem:hp_eproj_general_normalized} uniformly over
$t=0,\ldots,T-1$ with failure probability $\delta/T$ gives, with
probability at least $1-\delta$,}
$
\|e_t^{\mathrm{proj}}\|_2\le \zeta_t,
\qquad \forall t=0,\dots,T-1,
$
where $\zeta_t := C\sqrt{\frac{(d-1)\log\!\bigl(2dT/\delta\bigr)}{M}}
\Big(\|\tilde g_t\|_2+\sigma_t\Big)\quad+
C\,\frac{(d-1)\log\!\bigl(2dT/\delta\bigr)}{M}\,
(\|\tilde g_t\|_2 + B_t)$. Also, on $\mathcal E_T^{\mathrm{ZO}}$,
\[
\|e_t^{\mathrm{ZO}}\|_2\le \varepsilon,
\qquad \forall t=0,\dots,T-1.
\]
Hence, on $\mathcal E_T^{\mathrm{ZO}}$,
$
\|e_t^{\mathrm{tot}}\|_2
\le
\|e_t^{\mathrm{proj}}\|_2+\|e_t^{\mathrm{ZO}}\|_2
\le
\zeta_t+\varepsilon.
$

If condition~\eqref{eq:fixed_beta_condition} holds, then
$
\|e_t^{\mathrm{tot}}\|_2 \le \beta\|g_t\|_2,
\; \forall t=0,\dots,T-1,
$
and by Lemma~\ref{lem:federated_descent} when $K_{t}\in\mathcal{S}_c$ and with the choice of $\eta$, $K_{t+1}\in\mathcal{S}_c$ and
\[
J_{\mathrm{avg}}(K_{t+1})\le J_{\mathrm{avg}}(K_t)
\]
on $\mathcal E_T^{\mathrm{ZO}}$ with probability at least $1-\delta$.
Thus
\[
J_{\mathrm{avg}}(K_t)\le J_{\mathrm{avg}}(K_0)\le c
\qquad \forall t=0,\dots,T,
\]
so $K_t\in\mathcal S_c$ for all $t$ by induction, and all iterates are stabilizing.
\end{proof}

The parameter $\beta$ in Theorem~\ref{thm:stability_scalar_fixed_beta} may be
viewed as a conservative upper bound on the relative total error. In the ideal
exact-gradient regime, where $\varepsilon=0$ we have $e_t^{\mathrm{ZO}}=0$ and $\tilde g_t=g_t$. In this case, \eqref{eq:fixed_beta_condition} is satisfied if
\begin{equation}\label{eq:main-beta-lower-bound}
    \beta \gtrsim
\sqrt{\frac{d\log(2dT/\delta)}{M}}
+
\frac{d\log(2dT/\delta)}{M}.
\end{equation}
Thus, when the fleet size $M$ is sufficiently large relative to $d$, the
admissible value of $\beta$ becomes smaller. This, in turn, enlarges the
allowable stepsize range
$
0<\eta<\frac{2(1-\beta)}{L_c(1+\beta)^2},
$
and leads to faster convergence, highlighting a key large-scale advantage of
\textsc{ScalarFedLQR}. This is only a sufficient, and generally conservative, choice of $\beta$ for guaranteeing stability. Since it is derived from a uniform high-probability bound over the horizon $T$, it may overestimate the actual relative error in practice. As a result, the effective error level is often smaller, allowing for even larger stable stepsizes than those predicted by the theorem. The same qualitative dependence persists when the zeroth-order estimation and heterogeneity errors are sufficiently small.

\subsection{Linear convergence under a PL condition}

{Using the PL property in~\eqref{eq:individual_pl} together with
Lemma~\ref{lem:federated_descent}, and choosing
$
\eta^\star=\frac{1-\beta}{L_c(1+\beta)^2},
$
directly yields, with probability at least $1-\delta$,
\begin{align}
J_{\mathrm{avg}}(K_t)-J_{\mathrm{avg}}^\star
\le
\rho^t
\bigl(J_{\mathrm{avg}}(K_0)-J_{\mathrm{avg}}^\star\bigr),
\label{eq:linear_rate_scalar}
\end{align}
where
$
\rho
:=
1-\frac{\mu_c(1-\beta)^2}
{L_c(1+\beta)^2}
\in[0,1).
$
This follows by a standard argument; see \cite[Theorem~4.6]{bu_lqr_2019}.
Under the stated uniform relative gradient-error condition, \textsc{ScalarFedLQR} achieves linear convergence of the average cost to its optimal value; under fixed additive ZO error, the corresponding guarantee is to an accuracy neighborhood. The rate dependence on $M$ and $d$ follows through $\beta$. As discussed
above, $\beta$ scales, up to logarithmic factors, with $\sqrt{d/M}$.
Therefore, increasing the fleet size relative to the policy dimension reduces
the projection error, permits a larger stepsize, and improves the linear
convergence rate. Conversely, larger $d/M$ or less accurate local ZO
estimates lead to a larger effective $\beta$ and more conservative convergence.}

\section{Numerical Results}
\label{sec:numerical}

This section evaluates the performance of \textsc{ScalarFedLQR} in the model-free federated LQR setting. 
To ensure a fair and direct comparison, we adopt the same numerical setup and system-generation procedure as in the \FedLQR framework \cite{wang_model-free_2023}. 
In particular, all experiments use identical system dynamics, heterogeneity construction, and cost matrices, with the only difference being the learning algorithm and communication mechanism. {Although our theoretical analysis focuses on the homogeneous case, we also numerically
evaluate heterogeneous agents to demonstrate the empirical robustness of
\textsc{ScalarFedLQR}.}

\subsubsection{System Generation}
We consider a collection of $M$ heterogeneous discrete-time linear time-invariant
(LTI) systems of the form~\eqref{eq:dynamics}
where each system has state dimension $n_x=3$ and input dimension $n_u=3$.
Following the construction in~\cite{wang_model-free_2023}, the agent dynamics
are generated to satisfy a bounded heterogeneity condition. Specifically, a
nominal pair $(A_0,B_0)$ is fixed, from which the heterogeneous agent dynamics are generated by
structured perturbation:
\begin{equation}
A^{(n)} = A_0 + \gamma_1^{(n)} Z_1,
\qquad
B^{(n)} = B_0 + \gamma_2^{(n)} Z_2,
\end{equation}
where $Z_1,Z_2\in\mathbb R^{3\times 3}$ are fixed modification masks, and
$\gamma_1^{(n)}\sim \mathcal U(0,\epsilon_1)$ and
$\gamma_2^{(n)}\sim \mathcal U(0,\epsilon_2)$ are random perturbation levels.
The parameters $\epsilon_1,\epsilon_2>0$ control the degree of heterogeneity
across agents in~\eqref{eq::system_bound}. The nominal system itself is included in the population by
setting $(A^{(1)},B^{(1)})=(A_0,B_0)$.

The nominal system and cost matrices are given by
\(
A_0 =
\big(\begin{smallmatrix}
1.20 & 0.50 & 0.40 \\
0.01 & 0.75 & 0.30 \\
0.10 & 0.02 & 1.50
\end{smallmatrix}\big),
\;
B_0 = I_3,
\)
and
\(Q = 2I_3,
R = \tfrac{1}{2} I_3.\)
Throughout our simulations, we consider the following initial stabilizing controller $K_0 = 1.62 I_3$.
This control gain is used only to initialize a stabilizing policy for both algorithms, and is ensured to be stabilizing for all agents provided that $\epsilon_1$ and $\epsilon_2$ are small enough.

\begin{figure}[t]
    \centering

    \begin{subfigure}[t]{0.68\linewidth}
        \centering
        \includegraphics[width=\linewidth]{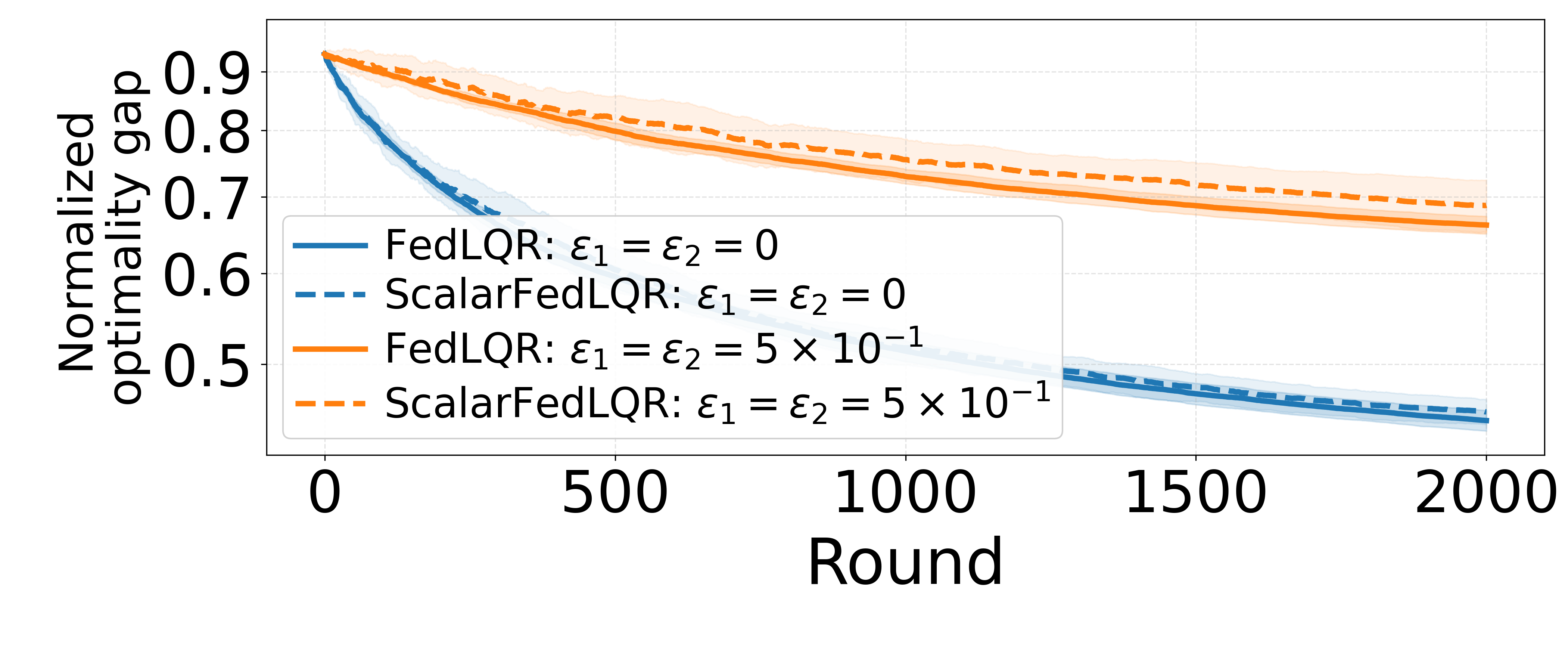}
        \caption{}
        \label{fig:optimality_gap_round}
    \end{subfigure}
    
    \begin{subfigure}[t]{0.68\linewidth}
        \centering
        \includegraphics[width=\linewidth]{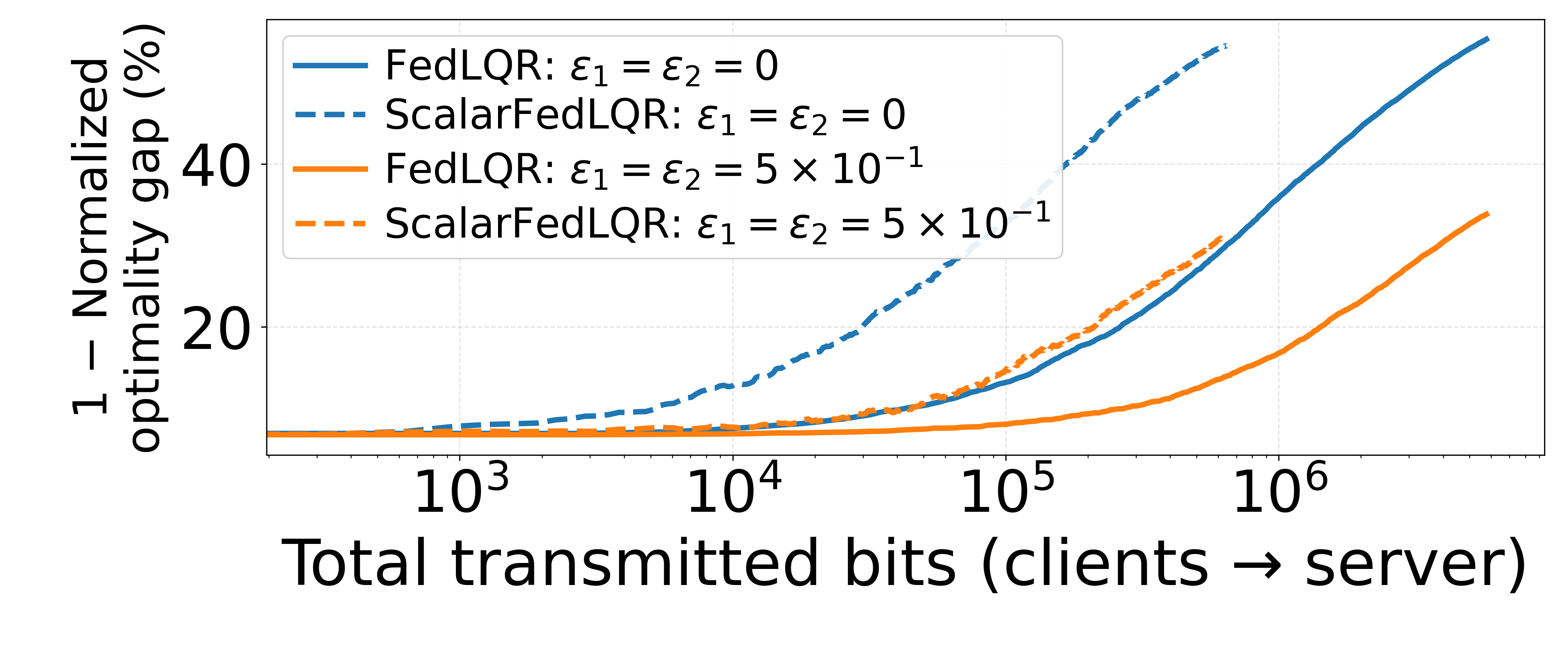}
        \caption{}
        \label{fig:recovery_vs_bits_all}
    \end{subfigure}

    \begin{subfigure}[t]{0.68\linewidth}
        \centering
        \includegraphics[width=\linewidth]{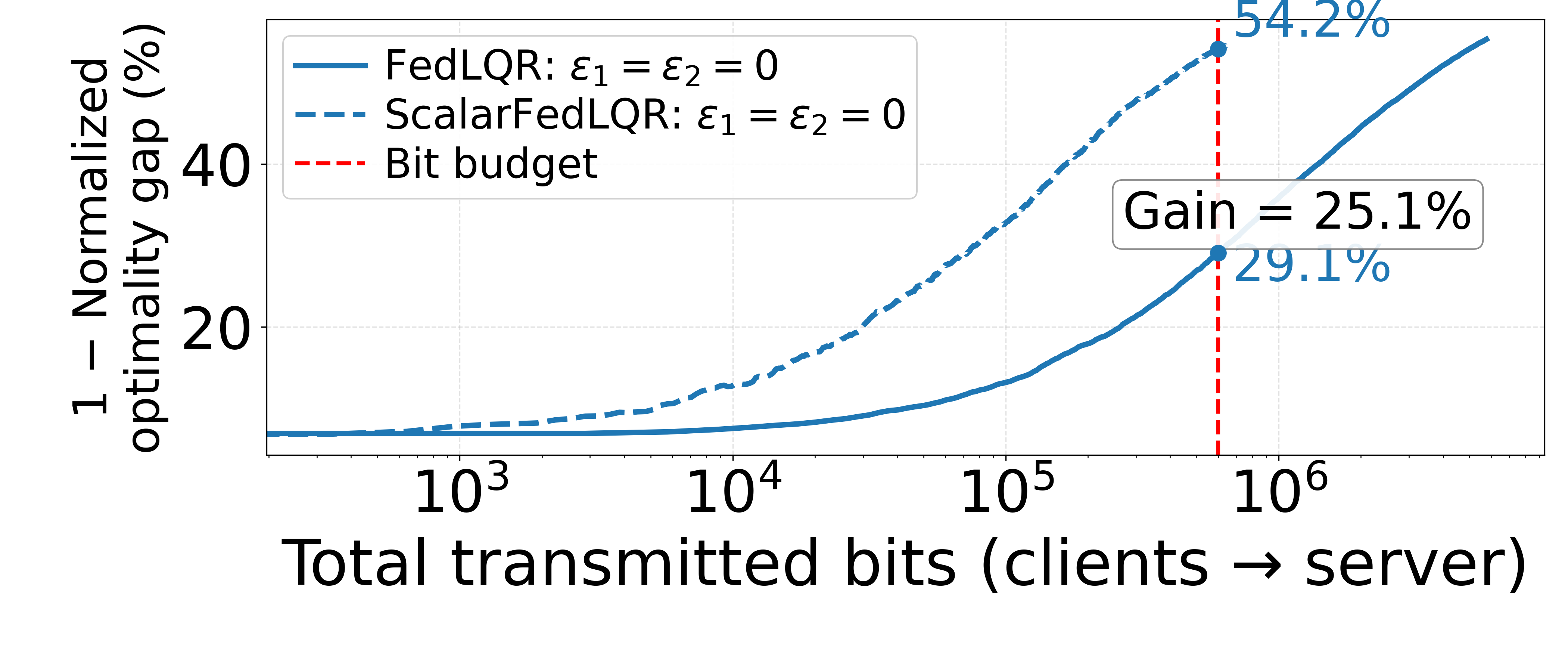}
        \caption{}
        \label{fig:budget_low_het}
    \end{subfigure}
    
    \begin{subfigure}[t]{0.68\linewidth}
        \centering
        \includegraphics[width=\linewidth]{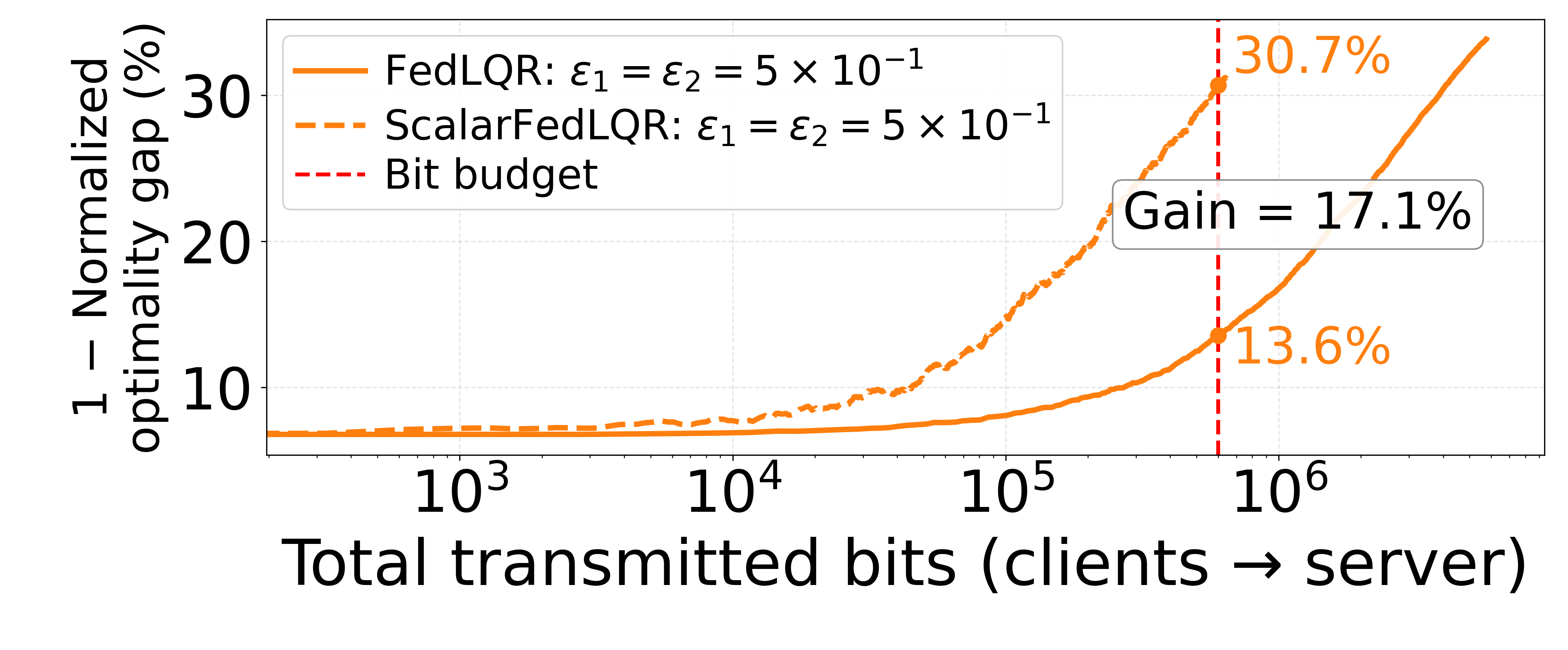}
        \caption{}
        \label{fig:budget_high_het}
    \end{subfigure}

    \caption{\small
    (a) Normalized optimality gap versus communication round.
    (b) Recovery percentage ($1-\text{normalized optimality gap}$)
    versus total transmitted bits under two heterogeneity levels.
    (c) Recovery percentage versus total transmitted bits for
    $\epsilon_1=\epsilon_2=0$ with a fixed bit budget of $6\times10^5$.
    (d) Recovery percentage versus total transmitted bits for
    $\epsilon_1=\epsilon_2=5\times10^{-1}$ with a fixed bit budget of
    $6\times10^5$.
    }
    \label{fig:convergence_communication_results}
    \vspace{-0.6cm}
\end{figure}

\subsubsection{Experimental Protocol}
\label{sec:experimental_protocol}
All agents perform model-free  PO using local trajectory
rollouts to estimate policy gradients. Unless otherwise stated, all methods are
evaluated under the same simulation and sampling configuration as in
\cite{wang_model-free_2023}. In particular, the number of agents is fixed to
$M=10$, the total number of communication rounds is $T=2000$, and the stepsize
is set to $\eta=0.01$ for both algorithms. For the zeroth-order gradient estimator, each local
gradient estimate is computed using $N_{\mathrm{traj}}=5$ trajectories, rollout
length $\tau=15$, and smoothing radius $r=0.1$. Each reported curve is averaged
over $10$ independent Monte Carlo runs.

To quantify communication, we count each transmitted scalar as a $32$-bit
floating-point number. Therefore, FedLQR incurs an uplink cost proportional to the full gradient dimension (here $n_u \times n_x = 9$), whereas \textsc{ScalarFedLQR} incurs only scalar-level communication per agent per round. 

{To assess performance, we consider the normalized optimality gap
$
\frac{J_{\mathrm{avg}}(K)-J_{\mathrm{avg}}(K^\star)}
     {J_{\mathrm{avg}}(K^\star)}.
$
For the homogeneous case, $K^\star$ is the common optimal controller obtained
from the DARE for $(A_0,B_0,Q,R)$. For the heterogeneous experiments, where
the exact population-optimal controller is generally unavailable, we use the
nominal-system optimum $K_1^\star$ as the reference and report
$
[J_1(K)-J_1(K_1^\star)]/J_1(K_1^\star).
$}

\subsubsection{Results}
{Fig.~\ref{fig:convergence_communication_results}(a) shows the convergence
behavior versus communication rounds under different heterogeneity levels.
Consistent with the theoretical analysis, \textsc{ScalarFedLQR} achieves
convergence behavior comparable to the \FedLQR baseline across the considered
heterogeneity levels.}

In order to also compare the communication requirement, we look at the percentages of cost improvement versus total number of bits transferred by each algorithm. The main advantage of \textsc{ScalarFedLQR} appears when performance is measured against communication cost. As shown in
Fig.~\ref{fig:recovery_vs_bits_all}, for a fixed number of transmitted bits, \textsc{ScalarFedLQR} consistently attains a higher recovery percentage $1-\text{normalized optimality gap}$ than \FedLQR$.$ This reflects the fact that \textsc{ScalarFedLQR} replaces full uplink gradient transmission by scalar communication, thereby using the available communication budget much more efficiently. The benefit is especially clear at the fixed budget of $6\times10^5$ bits. In the low heterogeneity setting, Fig.~\ref{fig:budget_low_het} shows that \textsc{ScalarFedLQR} achieves $54.2\%$ recovery, compared with $29.1\%$ for \FedLQR, which corresponds to a gain of $25.1$ percentage points. In the higher-heterogeneity setting, Fig.~\ref{fig:budget_high_het} shows that \textsc{ScalarFedLQR} still achieves $30.7\%$ recovery, compared with $13.6\%$ for \FedLQR, corresponding to a gain of $17.1$ percentage points.

Overall, these results show that \textsc{ScalarFedLQR} preserves performance comparable to \FedLQR when measured by communication rounds, while yielding a substantial reduction in communication cost and significantly higher recovery under a fixed bit budget regardless of heterogeneity levels. Finally, although the present experiments keep the model-free oracle fixed across methods, the results can be further improved by strengthening the zeroth-order gradient estimates, for example by increasing the number of rollouts or the trajectory length. We emphasize, however, that this is not the main focus of the paper; the central claim is the communication-efficiency gain achieved by scalar
uplink transmission. The simulation code is available online \cite{github_code}.

\section{Conclusion}
\label{sec:conclusion}

We introduced \textsc{ScalarFedLQR}, a communication-efficient federated algorithm
for model-free linear quadratic regulator (LQR) control with {homogeneous}
agents. By replacing full policy-gradient transmission with a single scalar
projection of each local zeroth-order gradient estimate, the proposed method
reduces the per-agent uplink communication cost from $\mathcal{O}(d)$ to
$\mathcal{O}(1)$, independently of the policy dimension. We showed that the
aggregated scalar-projection update defines a valid descent direction whose
approximation error improves with the number of participating agents. Under
standard stability and regularity conditions, we established that all iterates
remain stabilizing and that \textsc{ScalarFedLQR} achieves a linear convergence rate. Numerical experiments further confirmed that the
proposed method achieves performance comparable to full-gradient federated LQR
while significantly reducing communication. Future work includes sharpening the convergence analysis of \textsc{ScalarFedLQR} under more general heterogeneity and oracle conditions while preserving its communication-efficiency advantages.

\bibliographystyle{ieeetr}
\bibliography{references, bib/Reference, bib/alias}

\begin{thebibliography}{10}

\bibitem{hu_toward_2023}
B.~Hu, K.~Zhang, N.~Li, M.~Mesbahi, M.~Fazel, and T.~Başar, ``Toward a
  {Theoretical} {Foundation} of {Policy} {Optimization} for {Learning}
  {Control} {Policies},'' {\em Annual Review of Control, Robotics, and
  Autonomous Systems}, vol.~6, pp.~123--158, May 2023.

\bibitem{talebi_policy_2026}
S.~Talebi, Y.~Zheng, S.~Kraisler, N.~Li, and M.~Mesbahi, ``Policy
  {Optimization} in {Control}: {Geometry} and {Algorithmic} {Implications},''
  in {\em Encyclopedia of {Systems} and {Control} {Engineering} ({First}
  {Edition})} (Z.~Ding, ed.), pp.~39--61, Oxford: Elsevier, first edition~ed.,
  2026.

\bibitem{fazel_global_2018}
M.~Fazel, R.~Ge, S.~Kakade, and M.~Mesbahi, ``Global convergence of policy
  gradient methods for the linear quadratic regulator,'' in {\em Int. {Conf}.
  on {Machine} {Learning}}, pp.~1467--1476, PMLR, July 2018.

\bibitem{bu_lqr_2019}
J.~Bu, A.~Mesbahi, M.~Fazel, and M.~Mesbahi, ``{LQR} through the {Lens} of
  {First} {Order} {Methods}: {Discrete}-time {Case},'' July 2019.
\newblock arXiv:1907.08921 [cs, eess, math].

\bibitem{malik_derivative-free_2019}
D.~Malik, A.~Pananjady, K.~Bhatia, K.~Khamaru, P.~Bartlett, and M.~Wainwright,
  ``Derivative-free methods for policy optimization: Guarantees for linear
  quadratic systems,'' in {\em Proceedings of the 22nd International Conference
  on Artificial Intelligence and Statistics}, vol.~89, pp.~2916--2925, PMLR,
  16--18 Apr 2019.

\bibitem{dean_sample_2020}
S.~Dean, H.~Mania, N.~Matni, B.~Recht, and S.~Tu, ``On the sample complexity of
  the linear quadratic regulator,'' {\em Foundations of Computational
  Mathematics}, vol.~20, no.~4, pp.~633--679, 2020.

\bibitem{dulac-arnold_challenges_2019}
G.~Dulac-Arnold, D.~Mankowitz, and T.~Hester, ``Challenges of real-world
  reinforcement learning,'' 2019.
\newblock arXiv:1904.12901 [cs.LG].

\bibitem{mohammadi_learning_2020}
H.~Mohammadi, M.~R. Jovanovic, and M.~Soltanolkotabi, ``Learning the model-free
  linear quadratic regulator via random search,'' in {\em Learning for
  {Dynamics} and {Control}}, pp.~531--539, PMLR, 2020.

\bibitem{zhao_global_2023}
F.~Zhao, K.~You, and T.~Başar, ``Global convergence of policy gradient
  primal–dual methods for risk-constrained {LQRs},'' {\em IEEE Transactions
  on Automatic Control}, vol.~68, no.~5, pp.~2934--2949, 2023.

\bibitem{talebi_data-driven_2023}
S.~Talebi, A.~Taghvaei, and M.~Mesbahi, ``Data-driven {Optimal} {Filtering} for
  {Linear} {Systems} with {Unknown} {Noise} {Covariances},'' in {\em Advances
  in {Neural} {Information} {Processing} {Systems} ({NeurIPS})}, vol.~36,
  pp.~69546--69585, Curran Associates, Inc., 2023.

\bibitem{umenberger_globally_2022}
J.~Umenberger, M.~Simchowitz, J.~Perdomo, K.~Zhang, and R.~Tedrake, ``Globally
  {Convergent} {Policy} {Search} for {Output} {Estimation},'' in {\em Advances
  in {Neural} {Information} {Processing} {Systems}} (S.~Koyejo, S.~Mohamed,
  A.~Agarwal, D.~Belgrave, K.~Cho, and A.~Oh, eds.), vol.~35, pp.~22778--22790,
  Curran Associates, Inc., 2022.

\bibitem{alemzadeh_data-driven_2024}
S.~Alemzadeh, S.~Talebi, and M.~Mesbahi, ``Data-{Driven} {Structured} {Policy}
  {Iteration} for {Homogeneous} {Distributed} {Systems},'' {\em IEEE
  Transactions on Automatic Control}, vol.~69, pp.~5979--5994, Sept. 2024.

\bibitem{wang_model-free_2023}
H.~Wang, L.~F. Toso, A.~Mitra, and J.~Anderson, ``Model-free {Learning} with
  {Heterogeneous} {Dynamical} {Systems}: {A} {Federated} {LQR} {Approach},''
  Aug. 2023.
\newblock arXiv:2308.11743 [math].

\bibitem{qi_federated_2021}
J.~Qi, Q.~Zhou, L.~Lei, and K.~Zheng, ``Federated {Reinforcement} {Learning}:
  {Techniques}, {Applications}, and {Open} {Challenges},'' {\em Intelligence \&
  Robotics}, 2021.
\newblock arXiv:2108.11887 [cs].

\bibitem{mcmahan_communication-efficient_2017}
H.~B. McMahan, E.~Moore, D.~Ramage, S.~Hampson, and B.~A. y~Arcas,
  ``Communication-efficient learning of deep networks from decentralized
  data,'' in {\em Proceedings of the 20th International Conference on
  Artificial Intelligence and Statistics (AISTATS)}, pp.~1273--1282, 2017.

\bibitem{konevcny_federated_2016}
J.~Kone{\v{c}}n{\`y}, H.~B. McMahan, F.~X. Yu, P.~Richt{\'a}rik, A.~T. Suresh,
  and D.~Bacon, ``Federated learning: Strategies for improving communication
  efficiency,'' {\em arXiv preprint arXiv:1610.05492}, 2016.

\bibitem{kanakeri2025harnessing}
V.~Kanakeri, S.~Bajaj, A.~Verma, V.~Gupta, and A.~Mitra, ``Harnessing data from
  clustered lqr systems: Personalized and collaborative policy optimization,''
  {\em arXiv preprint arXiv:2511.17489}, 2025.

\bibitem{zhu_deep_2019}
L.~Zhu, Z.~Liu, and S.~Han, ``Deep leakage from gradients,'' in {\em Advances
  in Neural Information Processing Systems (NeurIPS)}, vol.~32,
  pp.~14774--14784, Curran Associates, Inc., 2019.

\bibitem{fournier2023can}
L.~Fournier, S.~Rivaud, E.~Belilovsky, M.~Eickenberg, and E.~Oyallon, ``Can
  forward gradient match backpropagation?,'' in {\em International Conference
  on Machine Learning}, pp.~10249--10264.

\bibitem{rostami2024projected}
M.~Rostami and S.~S. Kia, ``Projected forward gradient-guided frank-wolfe
  algorithm via variance reduction,'' {\em IEEE Control Systems Letters},
  vol.~8, pp.~3153--3158, 2024.

\bibitem{nesterov2017random}
Y.~Nesterov and V.~Spokoiny, ``Random gradient-free minimization of convex
  functions,'' {\em Foundations of Computational Mathematics}, vol.~17, no.~2,
  pp.~527--566, 2017.

\bibitem{MR-ST-SSK:26}
S.~T. Mohammadreza~Rostami and S.~S. Kia, ``Scalar federated learning for
  linear quadratic regulator,'' {\em arXiv:2604.05088}, 2026.

\bibitem{rostami2025fedmpdd}
M.~Rostami and S.~S. Kia, ``Fedmpdd: Communication-efficient federated learning
  with privacy preservation attributes via projected directional derivative,''
  {\em arXiv preprint arXiv:2512.20814}, 2025.

\bibitem{rostami2023federated}
M.~Rostami and S.~S. Kia, ``Federated learning using variance reduced
  stochastic gradient for probabilistically activated agents,'' 2023.

\bibitem{github_code}
M.~Rostami, S.~Talebi, and S.~S. Kia, ``Scalar federated learning for linear
  quadratic regulator,'' 2026.
\newblock https://github.com/RostamiHub/ScalarFedLQR.

\bibitem{tropp2012user}
J.~A. Tropp, ``User-friendly tail bounds for sums of random matrices,'' {\em
  Foundations of computational mathematics}, vol.~12, no.~4, pp.~389--434,
  2012.

\bibitem{tropp2015introduction}
J.~A. Tropp, ``An introduction to matrix concentration inequalities,'' {\em
  Foundations and trends{\textregistered} in machine learning}, vol.~8,
  no.~1-2, pp.~1--230, 2015.

\end{thebibliography}

\newpage
\appendices

\section{Technical Details}

\begin{proof}[Proof of Lemma~\ref{lem:federated_descent}]
If $g_t=0$, then the relative-error condition implies
$e_t^{\mathrm{tot}}=0$, and hence $\bar g_t=0$, so the result is
immediate. Suppose therefore that $g_t\neq0$ and define
$
K_t(s):=K_t-s\eta\,\bar g_t, ~ s\in[0,1].
$
For any $s$ such that $K_t(\tau)\in\mathcal S_c$ for all
$\tau\in[0,s]$, the $L_c$-smoothness of $J_{\mathrm{avg}}$ on
$\mathcal S_c$ gives
\begin{align}
J_{\mathrm{avg}}(K_t(s))
\le
J_{\mathrm{avg}}(K_t)
-s\eta\langle g_t,\bar g_t\rangle
+\frac{L_c s^2\eta^2}{2}\|\bar g_t\|_2^2 .
\label{eq:descent_pf_2}
\end{align}
Now, recall the scalar-projection aggregation rule of $\bar g_t$ as in
\eqref{eq:g_bar} and the local model-free gradient estimation of the
$n$-th agent $\tilde g_{t,n}$ as in \eqref{eq::g_t_n}. Define the total
error as
\[
e_t^{\mathrm{tot}}
:= \bar g_t-g_t
= (\bar g_t-\tilde g_t)+(\tilde g_t-g_t)
=: e_t^{\mathrm{proj}}+e_t^{\mathrm{ZO}}.
\]
Hence,
$
\bar g_t=g_t+e_t^{\mathrm{tot}}.
$
Substituting this identity into the inner-product term in
\eqref{eq:descent_pf_2}, we obtain
\begin{align}
\langle g_t,\bar g_t\rangle
=
\langle g_t,g_t+e_t^{\mathrm{tot}}\rangle
=
\|g_t\|_2^2+\langle g_t,e_t^{\mathrm{tot}}\rangle .
\label{eq:descent_pf_3}
\end{align}
Similarly,
\begin{align}
\|\bar g_t\|_2^2
=
\|g_t+e_t^{\mathrm{tot}}\|_2^2 .
\label{eq:descent_pf_4}
\end{align}
Plugging \eqref{eq:descent_pf_3} and \eqref{eq:descent_pf_4} into
\eqref{eq:descent_pf_2} yields
\begin{align}
J_{\mathrm{avg}}(K_t(s))
\le
J_{\mathrm{avg}}(K_t)
-s\eta\|g_t\|_2^2
-s\eta\langle g_t,e_t^{\mathrm{tot}}\rangle
\nonumber\\
+\frac{L_c s^2\eta^2}{2}
\|g_t+e_t^{\mathrm{tot}}\|_2^2 .
\label{eq:descent_pf_5}
\end{align}

Using Cauchy--Schwarz, we have
$
-\langle g_t,e_t^{\mathrm{tot}}\rangle
\le
\|g_t\|_2\,\|e_t^{\mathrm{tot}}\|_2,
$
and by the triangle inequality, we have
$
\|g_t+e_t^{\mathrm{tot}}\|_2
\le
\|g_t\|_2+\|e_t^{\mathrm{tot}}\|_2.
$
Therefore, from \eqref{eq:descent_pf_5},
\begin{align}
J_{\mathrm{avg}}(K_t(s))
\le
J_{\mathrm{avg}}(K_t)
-s\eta\|g_t\|_2^2
+s\eta\|g_t\|_2\,\|e_t^{\mathrm{tot}}\|_2
\nonumber\\
+\frac{L_c s^2\eta^2}{2}
\bigl(\|g_t\|_2+\|e_t^{\mathrm{tot}}\|_2\bigr)^2 .
\label{eq:descent_pf_6}
\end{align}

Now suppose that
$
\|e_t^{\mathrm{tot}}\|_2\le\beta_t\|g_t\|_2.
$
Then
$
s\eta\|g_t\|_2\,\|e_t^{\mathrm{tot}}\|_2
\le
s\eta\beta_t\|g_t\|_2^2,
$
and
$
\bigl(\|g_t\|_2+\|e_t^{\mathrm{tot}}\|_2\bigr)^2
\le
(1+\beta_t)^2\|g_t\|_2^2.
$
Substituting these bounds into \eqref{eq:descent_pf_6}, we arrive at
\begin{multline}
J_{\mathrm{avg}}(K_t(s))
\le
J_{\mathrm{avg}}(K_t)
\\-s\eta
\left[
(1-\beta_t)
-\frac{L_c s\eta}{2}(1+\beta_t)^2
\right]
\|g_t\|_2^2 .
\label{eq:segment_descent}
\end{multline}

Since
$
0<\eta<\frac{2(1-\beta_t)}
{L_c(1+\beta_t)^2},
$
the term in brackets in \eqref{eq:segment_descent} is positive for every
$s\in(0,1]$. Suppose, by contradiction, that the segment leaves
$\mathcal S_c$, and let $s^\star\in(0,1]$ denote its first exit point.
Then $K_t(s)\in\mathcal S_c$ for all $s\in[0,s^\star]$, while
$J_{\mathrm{avg}}(K_t(s^\star))=c$. However,
\eqref{eq:segment_descent} gives
\[
J_{\mathrm{avg}}(K_t(s^\star))
<
J_{\mathrm{avg}}(K_t)
\le c,
\]
which is a contradiction. Hence, $K_t(s)\in\mathcal S_c$ for all
$s\in[0,1]$. In particular, $K_{t+1}=K_t(1)\in\mathcal S_c$, and
setting $s=1$ in \eqref{eq:segment_descent} yields
\begin{align}
J_{\mathrm{avg}}(K_{t+1})
\le
J_{\mathrm{avg}}(K_t)
-\eta
\left[
(1-\beta_t)-\frac{L_c\eta}{2}(1+\beta_t)^2
\right]
\|g_t\|_2^2 .
\end{align}
This proves the desired one-step descent bound.
\end{proof}

\begin{proof}[Proof of Lemma~\ref{lem:hp_eproj_general_normalized}]
Let $\Delta_{t,n}:=\tilde g_{t,n}-\tilde g_t$, so that $\sum_{n=1}^{M}\Delta_{t,n}=0$. Then
\begin{align*}
e^{\mathrm{proj}}_t
\!\!=\!\!
\underbrace{\Big(\frac{1}{M}\sum\nolimits_{n=1}^{M}(d\,v_{t,n}v_{t,n}^\top-I)\Big)\tilde g_t}_{T_1}
+\nonumber\\
\underbrace{\frac{1}{M}\sum\nolimits_{n=1}^{M}(d\,v_{t,n}v_{t,n}^\top-I)\Delta_{t,n}}_{T_2}.
\end{align*} 
For the first term, write
\[
T_1=\frac{1}{M}\sum\nolimits_{n=1}^{M}X_n,
\qquad
X_n := (d\,v_{t,n}v_{t,n}^\top-I)\tilde g_t .
\]
Conditioned on $\tilde g_t$, the vectors $\{X_n\}_{n=1}^M$ are independent and mean-zero since
$
\mathbb E[d\,v_{t,n}v_{t,n}^\top]=I.
$
Now define $u_{t,n}:=\sqrt d\,v_{t,n}\in\{\pm1\}^d$. Then
$
d\,v_{t,n}v_{t,n}^\top = u_{t,n}u_{t,n}^\top.
$
Hence
\[
\|d\,v_{t,n}v_{t,n}^\top-I\|_{\mathrm{op}}
=
\|u_{t,n}u_{t,n}^\top-I\|_{\mathrm{op}}
=
d-1.
\]
Therefore,
$
\|X_n\|_2
\le
(d-1)\|\tilde g_t\|_2
=:R_1.$
Also, using again $u_{t,n}=\sqrt d\,v_{t,n}$ and 
\[
\mathbb E\|(d\,v_{t,n}v_{t,n}^\top-I)z\|_2^2
=
\mathbb E\|(u_{t,n}u_{t,n}^\top-I)z\|_2^2
=
(d-1)\|z\|_2^2 .
\]
Thus,
\[
\sum\nolimits_{n=1}^{M}\mathbb E\|X_n\|_2^2
=
M(d-1)\|\tilde g_t\|_2^2
=:V_1.
\]

Applying a vector Bernstein inequality gives, with probability at least $1-\delta/2$ \cite{tropp2012user}\cite[Corollary 6.2.1]{tropp2015introduction},
\begin{align}
\label{eq:T1_bound_normalized}
\|T_1\|_2
\le C\,\|\tilde g_t\|_2
\sqrt{\frac{(d-1)\log(2d/\delta)}{M}}
\nonumber\\
+\,
C\,\|\tilde g_t\|_2
\frac{(d-1)\log(2d/\delta)}{M}.
\end{align}
for an absolute constant $C>0$. For the second term, write
$
T_2=\frac{1}{M}\sum_{n=1}^{M}Y_n,
\qquad
Y_n := (d\,v_{t,n}v_{t,n}^\top-I)\Delta_{t,n}.
$
Conditioned on $\{\Delta_{t,n}\}_{n=1}^{M}$, the vectors $\{Y_n\}_{n=1}^{M}$ are independent and mean-zero. Similar to the first term we have 
\[
\|Y_n\|_2
\le
(d-1)\|\Delta_{t,n}\|_2
\le
(d-1)B_t
=:R.
\]
Also,
$
\sum_{n=1}^{M}\mathbb E\|Y_n\|_2^2
=
(d-1)\sum_{n=1}^{M}\|\Delta_{t,n}\|_2^2
\le
(d-1)\,M\,\sigma_t^2
=:V.
$
Hence, similarly by vector Bernstein, with probability at least $1-\delta/2$,
\begin{equation}\label{eq:T2_bound_normalized}
\|T_2\|_2
\!\le\!
C\!\!\left(
\!\sigma_t\sqrt{\frac{(d-1)\log(2d/\delta)}{M}}
\!+
\!B_t\frac{(d-1)\log(2d/\delta)}{M}
\!\!\right),
\end{equation}
for an absolute constant $C>0$. Finally, a union bound over the events \eqref{eq:T1_bound_normalized} and
\eqref{eq:T2_bound_normalized} yields probability at least $1-\delta$.
Summing the two bounds proves \eqref{eq:hp_eproj_general_normalized}.
\end{proof}

Lemma~\ref{lem:hp_eproj_general_normalized} controls $e_t^{\mathrm{proj}}$ at a
single round. To establish stability over the full optimization horizon, we next
strengthen this result to a uniform high-probability bound that holds
simultaneously for all rounds $t=0,\dots,T-1$.

\end{document}